\documentclass[11pt]{article}

\usepackage{enumitem}
\newlist{desiderata}{enumerate}{1}
\setlist[desiderata,1]{
  label={Desiderata~\arabic*},
  leftmargin=10\parindent,
  align=right,
  labelindent=4\parindent,
  labelsep=5mm
}

\usepackage{amsmath}
\usepackage{amssymb}
\usepackage{booktabs}
\usepackage{bm}
\usepackage{multirow}
\usepackage{tikz}
\usepackage{pgfplots}

\newcommand{\pt}[2]{\frac{\partial{#1}}{\partial{#2}}}

\newenvironment{formulation2}{\fontsize{10}{11}\selectfont}{\par}

\usepackage{hyperref}

\usepackage{amsthm}

\hypersetup{colorlinks=true,pdftitle=""}

\newtheorem{conj}{Conjecture}[section]
\newtheorem{theorem}[conj]{Theorem}

\newtheorem{proposition}[conj]{Proposition}

\begin{document}


\title{A Closed-Form EVSI Expression for a Multinomial Data-Generating Process}

\author{%
Adam Fleischhacker\thanks{Department of Business Administration, University of Delaware, Newark, DE 19716, email: {\tt ajf@udel.edu}}, 
Pak-Wing Fok\thanks{Department of Mathematical Sciences, University of Delaware, Newark, DE 19716, email: {\tt pakwing@udel.edu}},
Mokshay Madiman\thanks{Department of Mathematical Sciences, University of Delaware, Newark, DE 19716, email: {\tt madiman@udel.edu}},
Nan Wu\thanks{Institute for Financial Services Analytics, University of Delaware, Newark, DE 19716, email: {\tt nanw@udel.edu}} 
} 

\maketitle

\begin{abstract}
This paper derives analytic expressions for the expected value of sample information (EVSI), the expected value of distribution information (EVDI), and the optimal sample size when data consists of independent draws from a bounded sequence of integers.  Due to challenges of creating tractable EVSI expressions, most existing work valuing data does so in one of three ways: 1) analytically through closed-form expressions on the \textit{upper bound} of the value of data, 2) calculating the expected value of data using numerical comparisons of decisions made using simulated data to optimal decisions where the underlying data distribution is known, or 3) using variance reduction as proxy for the uncertainty reduction that accompanies more data.   For the very flexible case of modelling integer-valued observations using a multinomial data-generating process with Dirichlet prior, this paper develops expressions that 1) generalize existing beta-Binomial computations,   2) do not require prior knowledge of some underlying ``true'' distribution, and 3) can be computed prior to the collection of any sample data.  
\end{abstract}



%


\section{Introduction}

The seminal work of \cite{raiffa1961} introduced preposterior analysis, a Bayesian recipe for estimating the value of information (VOI) prior to knowing the information's content.  The \emph{expected value of sample information} (EVSI), a particularly valuable VOI computation, values the information contained in sample observations prior to their collection.  \cite{raiffa1961} include many closed-form and oft-used expressions for calculating EVSI under the assumption of quadratic loss.  One such expression is for a Bernoulli data-generating process with beta prior distribution (a.k.a. a beta-Binomial model); each observation being either zero or one \cite[Table 6.2, p. 191]{raiffa1961}.  In this paper, we generalize the beta-binomial EVSI expression beyond binary-valued observations to the case where each data point is drawn from a bounded sequence of integers.  These results expand the availability of tractable VOI expressions to a useful scenario where previously value could only be approximated or bounded when a closed-form expression was needed.


Depending on a modeler's choices of actions, states of uncertainty, loss (or utility) functions, and probability models, tractable calculations of VOI may exist, but intractable formulations, especially for EVSI, are much more common.  In fact, reputed statistician Dennis Lindley has remarked that the question of sample size ``is embarrassingly difficult to answer'' due to difficulties calculating EVSI \cite{lindley1997}. More generally, \cite{hilton1981determinants} shows that simply characterizing the relationship between information and value is challenging; \cite{hilton1981determinants}'s work dispels the idea that information value will reliably exhibit monotonic relationships with information value determinants such as action flexibility, risk aversion, or a decision maker's wealth.

While for some EVSI and VOI problems, closed-form solutions are attainable \cite{raiffa1961, bickel2008relationship, bhattacharjya2013value}, value of information solutions are often difficult to formulate.  Hence, many papers are known for their ability to characterize aspects of VOI expressions such as the distributional properties of the expected value of perfect information (EVPI) \cite{mehrez1982some}, the impact of an exogenous variable on EVPI \cite{keisler2004comparative}, and the additivity of information value when multiple sources of uncertainty exist \cite{keisler2005additivity}. EVSI calculations, in particular, often result in intractable expressions of multiple integrals where only numerical methods can yield results \cite{lin1974}.  Even then, many numerical methods still require further simplifying assumptions (see, e.g., \cite{strong2015}).   While it is possible to approximate VOI computations via normal approximations (see, e.g., \cite{morita2008determining,jalal2018gaussian})  or using a computationally intense simulation-based methodology (see, e.g., \cite{fleischhacker2015,wang2021blockchain}), closed-form expressions yield instantaneous and accurate value computations with more interpretable insights regarding the effects of prior beliefs and sample sizes.  

In this paper, we provide a new EVSI calculation for a flexible (i.e. multinomial) data-generating process that adheres to three desiderata outlined in \cite[p.44]{raiffa1961}:
\begin{desiderata}[]
\item [\textbf{Tractable}] EVSI is easily calculated using a closed-form expression.
\item [\textbf{Rich}] A decision maker's prior beliefs and information are readily incorporated as part of the calculation.
\item [\textbf{Interpretable}] The expression for EVSI provides insight as to the effects of prior beliefs and sample size choices on the expected value of a sample.
\end{desiderata}

\begin{table}[h]
\centering
\renewcommand{\arraystretch}{0.8}
\begin{tabular}{ccc}
\toprule
{\scriptsize \textbf{Generating Process}} & {\scriptsize \textbf{Conjugate Prior}}  & {\scriptsize \textbf{Source}}\\
\midrule
\multirow{2}{*}{Bernoulli$(\theta)$}  & \multirow{2}{*}{$(\theta) \thicksim $Beta}     & {\cite{raiffa1961}} \\
& & \cite{phamgia1992}\\
\midrule
{Poisson$(\lambda)$}    & {$\lambda \thicksim$gamma}  & {\cite{raiffa1961}}\\
\midrule
\multirow{3}{*}{Normal$(\mu,\sigma)$}     & $\mu \thicksim$ Normal, $\sigma$ known       & {\cite{raiffa1961}} \\ \cline{2-3}
                            & $\mu$ known, ${\sigma^2} \thicksim$ inv. Gamma                & {\cite{raiffa1961}} \\\cline{2-3}
                            & ${\sigma^2} \thicksim$ inv. Gamma, $\mu | \sigma^2 \thicksim$ Normal  & {\cite{raiffa1961}}  \\
\midrule
{Multinomial($t$)\footnotemark}  & {$t \thicksim$ Dirichlet}     &  {This Paper} \\
\bottomrule
\end{tabular}
\caption{Position of this paper in comparison to other tractable EVSI calculations.}
\label{tab:EVSIcomparison}
\end{table}

Shown in Table 1, our point of departure is generalizing the EVSI calculation for a Bernoulli data-generating process with beta prior (a.k.a. a beta-binomial model) to the case of a multinomial data generating process with Dirichlet prior.  Rich treatment and illustrative examples surrounding EVSI calculations for the beta-binomial conjugacy can be found in \cite{howard1970decision}.  Additionally, \cite{phamgia1992} provide explicit closed-form value of information computations for the beta-binomial case and is very close in spirit to this work, but does not investigate the Dirichlet-multinomial setting. In relation to the multinomial sampling process we explore in this paper, existing work has focused on non-utility based approaches where data is valued based on its ability to bound a parameter of interest within a certain level of precision \cite{adcock1993,cao2009comparison}.  Our approach, in contrast, extends the utility-based valuation of sampling to a multinomial sampling environment to yield closed-form expressions for both EVSI and the expected value of distribution information (EVDI).  Publication of analytically tractable expressions will be able to supplant the still-present usage of Monte Carlo simulation in multinomial settings (see, e.g., \cite{xiang2019optimal}). 

\footnotetext{With support interpreted as a sequence of integer values.}

When closed-form EVSI expressions are unavailable, quantification of value created through uncertainty reduction typically relies on one of three techniques: 1) closed-form expressions on the \textit{upper bound} of the value of data, 2) simulated comparisons between valuing decisions made by an oracle who knows the underlying data distribution to decisions made by a less-informed decision maker, or 3) using variance-reduction as a proxy for how data reduces underlying uncertainty in the data-generating process.  For examples of the first type, \cite{mehrez1985effect} bound EVPI for a risk-averse decision maker and \cite{Yuechen2006} place an upper bound on the value of knowing the true distribution when one already knows the mean and variance of that distribution.  Examples of the second type often compare a Bayesian updating procedure to a known optimal solution \cite{eppen97, milner2005, chen2008dynamic, saghafian2016newsvendor}.  Lastly, computing the value of variance reduction independent of the specific quantity of data is also seen within the literature \cite{Gerchak1992, Kwak2011271}.

\section{Problem Setup}
\label{sec:PS}
Despite substantial efforts, notation for preposterior analysis has not been standardized and is often a matter of personal taste \cite{raiffa2008early}.  To aid the reader with this paper's notation surrounding its random variables and their realizations,  we present the following summary breaking the notation into three levels of analysis:

\bigskip

\begin{formulation2}
{\noindent{\bf 1. Data/Sample}. Data is an integer-valued random variable with support $\{0,1,\ldots,M\}$.  Sample is a random vector referring to either a sequence of $n$ data observations or a vector of counts representing the number of occurrences of each potential data value recorded in $n$ observations.}\\
\begin{tabular}{rl}
&\\
$D$:         &  A random variable representing a single data observation.\\
$d$:       &  A single realization of $D$ with integer valued support: $d \in \{0,1,\ldots,M\}$.\\
$X \equiv (X_1,\ldots,X_n)$:       &  A random \emph{vector} of $n$ observations of $D$.\\
$x \equiv (x_1,\ldots,x_n)$:       &  A realization of data vector $X$.\\
${\cal D}^n$:    &  The support of $X$ when $n$ realizations are observed.\\
$n_k$:       & the number of times that $k \in \{0,1,\ldots,M\}$ appears in $x$.\\
$(n_0,n_1,\ldots,n_M)$: &  A vector of counts of occurrences for each potential data value.
\end{tabular}
\end{formulation2}
\bigskip
\begin{formulation2}
{\noindent{\bf 2. Data/Sampling Distributions}. Data and sampling distributions are identical terms referring to the probability distribution governing the data-generating process.  Data distribution refers to generating individual data points and sampling distribution preferred when talking about a sequence of observations.} \\
\begin{tabular}{rl}
&\\
$T \equiv (T_0,T_1,\ldots,T_M)$:       &  A random vector representing a data distribution. \\
& Random elements $T_k$ are data distribution parameters \\
& representing the probability of data realization being $k$.\\
$t \equiv (t_0,t_1,\ldots,t_M)$:       &  A realization of random vector $T$ such that \\
& $t_k = p(D=k)$ for $k \in \{0,1,\ldots,M\}$.\\
$t^*$:       & The ``true'' data distribution or sampling distribution; \\
&only knowable by an oracle. \\
${\cal T}$:         &  The space or set of all possible data distributions. $T, t, t^* \in {\cal T}$. \\
\end{tabular}
\end{formulation2}
\bigskip
\begin{formulation2}
{\noindent{\bf 3. Prior/Posterior Distributions}. Continuous multivariate probability distributions with domain of all possible data distributions.} \\
\begin{tabular}{rl}
&\\
$\pi$:       &  A prior from which data distributions are generated.\\
$\pi_X$:       &  A posterior that updates $\pi$ in light of data $X$.\\
\end{tabular}
\end{formulation2}
\bigskip

\subsection{Modelling Data and Loss}
\label{sec:dataLoss}
Consider a \emph{data-generating process} that generates independent and identically distributed samples from a bounded sequence of $M+1$ integers.  For notational simplicity,
we rescale the sequence to be $[M] \equiv \{0,1,\ldots,M\}$.  For practical motivation, the data could represent product demand and the goal is to make accurate predictions for inventory control \cite{yelland2009bayesian}. For the specific case of demand uncertainty, we note that there are asymmetric and other loss functions that would be preferred to the quadratic loss function used here, but closed-form expressions are not forthcoming for those cases.

The data-generating process is governed by an unknown \emph{data distribution}, $t$, with discrete-finite support 
$[M]$.
Thus the statistical model for the data-generating process is parameterized by the standard $M$-dimensional simplex
of probabilities
$$
{\cal T} =\{t=(t_0,\ldots, t_M)\in\mathbb{R}_{+}^{M+1}: t_0+\ldots+t_M=1\};
$$
this infinite (but finite-dimensional) parameter space describes how we are labeling the potential data distributions.  If the sample size of the data is $n$, we have $n$ values $x_1, \ldots, x_n \in [M]$ being generated by the
data-generating process. For a given $t \in {\cal T}$, the associated data-generating process $p_t^{(n)}$ assigns probability
\begin{equation}
\label{eq:dataProc}
p_t^{(n)}(x_1,\ldots, x_n)=\prod_{i=1}^n t_{x_i}
\end{equation}
to this particular sequence of data values. In particular, if the sample size is 1, the data-generating process is simply given by
$$
p_t(d) \equiv p_t^{(1)}(d)=t_{d}, \qquad d \in [M].
$$

It is clear that the number of occurrences of particular data values in the sample is a sufficient statistic for the
model described, and that the sampling distribution for this sufficient statistic is just the multinomial model.
Specifically, if $n_d=|\{1\leq i\leq n: x_i=d\}|$, then $(n_0, \ldots, n_M)$ is a sufficient statistic, and we have, with
obvious abuse of notation,
\begin{equation}
\label{eq:sampProc}
p_t (n_0, \ldots, n_M)= \binom{n}{n_0 \cdots n_M} \prod_{d=0}^M t_d^{n_d}.
\end{equation}
Note that  $n_0+\ldots+n_M=n$ by definition; so we do not write the superscript $(n)$ when using the sufficient statistic to represent the data.


When making predictions for future data, ideally the action (or prediction) is close to the actual data realization.  For tractability, we consider a quadratic terminal opportunity \emph{loss function} for a single prediction to be of the following form:
\begin{equation}
\label{eq:quadLoss}
\ell(d, a) = k (d - a)^2
\end{equation}
where $k>0$ is a known constant, $a$ is the action/prediction, and $d \in [M]$ is the actual data realization.

To briefly make the above notation more concrete, let's imagine forecasting product demand for a product that will sell between 0 and 5 units ($M=5$).  Each period's i.i.d demand, $d \in \{0, 1, \ldots, 5\}$, has an associated probability of occurrence, $p_t(0), p_t(1), \ldots, p_t(5)$, which is represented more compactly as $t_0, t_1, \ldots, t_5$.  The effectiveness of any action will be measured using quadratic loss scaled by a factor $k$ such that if $k=5$, $d=4$, and $a=1$, then $\ell(4, 1) = 45$.  The decision maker is contemplating the value of $n=3$ observations where generated data, $(x_1,x_2,x_3)$, might be something like $(0, 5, 0)$ and the associated sufficient statistic of counts, $(n_0, \ldots, n_5)$, would be $(2, 0, 0, 0, 0, 1)$.  Note that $t \equiv t_0, t_1, \ldots, t_5$ parameterizes both the data-generating process of eq. (\ref{eq:dataProc}) yielding $(x_1,x_2,x_3)$ and the equivalent sampling process of eq. (\ref{eq:sampProc}) yielding $(n_0, \ldots, n_5)$.  As a result, we refer to $t$ as both data distribution and sampling distribution depending on context.

\subsection{Preposterior Analysis}
\label{subsec:oracle}
For any data distribution $t$, define the expectation of loss as:
\begin{equation}
R(t,a) = \mathbb{E}_{D|T=t}\left[ \ell(D,a)\right] = \sum_{d=0}^{M} p_{t}(d) \ell(d,a).
\end{equation}
where $R(t,a)$ is known as the Bayes risk.  Since a decision maker (DM) does not know the underlying ``true'' $t^*\in{\cal T}$ data distribution, the minimum Bayes risk, $\min_a R(t^*, a)$, is likely unachievable.

For a DM, risk is evaluated on an average basis based on the probability distribution the DM places over the simplex ${\cal T}$.  Without any sample observations, this distribution is the prior $\pi$ over all possible data distributions in ${\cal T}$.  The average risk of taking action $a$ using prior $\pi$ is
\begin{equation}
\bar{R}(\pi,a) = \mathbb{E}_{T} \left[ R(T,a) \right],
\label{eqn:evsiA}
\end{equation}
with $T \sim \pi$.
The Bayes action for $\pi$ is
\begin{equation}
a^*(\pi)=\underset{a \in {\cal A}}{\rm arg~min} ~\bar{R}(\pi,a).
\end{equation}
The Bayes risk for $\pi$ is
\begin{equation}
\bar{R}(\pi, a^*(\pi)) = \underset{a \in {\cal A}}{\rm min} ~\bar{R}(\pi,a).
\end{equation}

Access to a sample $X \equiv (X_1,\ldots,X_n)$ results in a different decision with different risk.  With sample observations, the DM applies Bayes' rule to update $\pi$ to $\pi_X$ (the posterior) and calculates the associated optimal Bayes action $a^*(\pi_X)$.  Since $X$ is unknown prior to actually collecting the sample, the Bayes risk for $\pi_X$ is itself a random variable.  Hence, we evaluate the DM's prior expectation of loss with sample information over all possible samples $X$,
\begin{eqnarray}
\mathbb{E}_X\left[\bar{R}(\pi_X,a^*(\pi_X))\right] = \mathbb{E}_T \mathbb{E}_{X|T} \left[R(T,a^*(\pi_X))\right],
\label{eqn:evsiB}
\end{eqnarray}
with $T \sim \pi$ and the right-hand side expression derived by substituting $\pi_X$ for $\pi$ in eq. (\ref{eqn:evsiA}) and applying the law of total expectation.

Thus, the \emph{expected value of a sample of information} (EVSI), $V_n(\pi)$, is the difference between the prior expectations of loss with and without sample $X$ under prior $\pi$:
\begin{eqnarray}
V_n(\pi)&=& \bar{R}(\pi,a^*(\pi)) - \mathbb{E}_X\left[\bar{R}(\pi_X,a^*(\pi_X))\right] \\
&=& \mathbb{E}_{T} \left[R(T,a^*(\pi))\right] - \mathbb{E}_T \mathbb{E}_{X|T} \left[R(T,a^*(\pi_X)) \right]
\label{eqn:evsi}
\end{eqnarray}
where $T \sim \pi$ and eq. (\ref{eqn:evsi}) follows from eqs. (\ref{eqn:evsiA}) and (\ref{eqn:evsiB}).  Proposition~\ref{th:nonnegativity} formalizes our intuition that this expected value of sample information should be non-negative.

\begin{proposition}
\label{th:nonnegativity}
Suppose data distribution $T \equiv (T_0,\ldots,T_M)$ is
drawn from a given prior $\pi$. Assume further that a DM
is given $n$ samples $X \equiv (X_1,\ldots,X_n)$ and updates
his/her prior to the posterior $\pi_X$. Then, under quadratic loss, the 
expected value of these $n$ samples is non-negative, i.e.
\begin{equation}
V_n(\pi) = \mathbb{E}_{T} \left[R(T,a^*(\pi))\right] - \mathbb{E}_T \mathbb{E}_{X|T} \left[R(T,a^*(\pi_X)) \right] \geq 0. \label{eqn:barV}
\end{equation}
\end{proposition}

\begin{proof}   
See Appendix.  \hfill$\Box$
\end{proof}   
\noindent Because the ordering within the sample $X$ does not matter, the inner expectation in (\ref{eqn:barV}) is performed over
$(n_0,n_1,\ldots,n_M) \sim {\rm Multinomial}(t)$ conditioned on $T=t$ where $n_j$ is the number of times
that $j \in [M]$ appears in the sample, and the outer expectation is performed
over $T \sim \pi$.


\section{Tractable Valuation of Sample Information}
To arrive at a tractable valuation for (\ref{eqn:evsi}), we leverage the Dirichlet distribution as a prior for three reasons: 1) it is a conjugate prior to categorical/multinomial outcomes, 2) its support is the $M$-dimensional simplex ${\cal T}$, and 3) it has flexibility to model many types of 
prior information for the decision maker.  With the Dirichlet assumption, the main result of this paper, Theorem~\ref{th:evsi}, can be presented:
%
%
\begin{theorem}
For data distribution $T$ with support $[M]$ and prior $\pi = \textrm{Dirichlet}(\alpha_0,\alpha_1,\ldots,\alpha_M)$, the expected reduction in quadratic loss after observing $n$ data samples, also called the expected value of sample information (\textit{EVSI}), is given by:
\begin{equation}
V_n(\pi) = \frac{k n (c_2-c_1^2) }{(n + \alpha)(1+\alpha)}. \label{eqn:EVSI}
\end{equation}
where $\alpha = \sum_{d=0}^M \alpha_d$ is the precision/concentration parameter of the Dirichlet distribution (see \cite{huang2005maximum}) and $c_{1} = \frac{1}{\alpha}\sum_{d=0}^M d \alpha_d$ and $c_{2} = \frac{1}{\alpha}\sum_{d=0}^M d^2 \alpha_d$ are the first and second moments of
the data under the marginal likelihood $(\alpha_1,\alpha_2,\ldots,\alpha_M)/\alpha$.
\label{th:evsi}
\end{theorem}
\begin{proof}    See Appendix. \hfill$\Box$
\end{proof}   

Theorem~\ref{th:evsi} gives the expected value of observing an $n$-trial multinomial sample with Dirichlet prior where support of the underlying data-generating process is the bounded sequence of integers $[M] = \{0,1,\ldots,M\}$.  This is a natural generalization of valuing an $n$-trial binomial sample with beta prior where support of the underlying data-generating process is restricted such that $[M] = \{0,1\}$.  With just a slight change of notation, we know from \cite{phamgia1992} that EVSI for the beta-binomial case in closed-form is:
\begin{equation}
\frac{k n}{n + \alpha_0 + \alpha_1} \cdot \frac{\alpha_0 \alpha_1}{(\alpha_0 + \alpha_1)^2(\alpha_0 + \alpha_1 + 1)} \label{eqn:EVSIphamgia}
\end{equation}
where $\pi \sim Beta(\alpha_0,\alpha_1)$.  Replacing this prior with the equivalent Dirichlet parameterization of $\pi \sim Dirichlet(\alpha_0,\alpha_1)$ and using Theorem~\ref{th:evsi} yields an identical result:
\begin{equation}
\label{eqn:evsiComparison}
  \begin{split}
  V_n(\pi) &= \frac{k n (c_2-c_1^2) }{(n + \alpha)(1+\alpha)} \\
  &= \frac{k n}{(n + \alpha_0 + \alpha_1)}\cdot\frac{\frac{\alpha_1}{\alpha_0 + \alpha_1}-\frac{\alpha_1^2}{(\alpha_0 + \alpha_1)^2}}{(\alpha_0 + \alpha_1 + 1)} \\
  &= \frac{k n}{(n + \alpha_0 + \alpha_1)}\cdot\frac{\alpha_0 \alpha_1}{(\alpha_0 + \alpha_1)^2(\alpha_0 + \alpha_1 + 1)} 
  \end{split}
\end{equation}

As a direct consequence of Theorem \ref{th:evsi}, when $n \to \infty$, we have an expression for the \textit{expected value of distribution information} (EVDI), as an infinite sample gives the data distribution exactly:
\begin{equation}
V_{\infty}(\pi) = \lim_{n \to \infty} V_n(\pi) = \frac{k(c_2 - c_1^2)}{1+\alpha}. \label{eqn:EVDI}
\end{equation}
Lastly, we can express the efficiency $\eta$ of the sample information as a function of the number of sample points using the ratio of (\ref{eqn:EVSI}) to (\ref{eqn:EVDI}) as:
\begin{equation}
\eta = \frac{n}{n + \alpha}. \label{eqn:eta}
\end{equation}
Hence, the percentage of value obtained through sampling is given by the ratio of the number of data points $n$ to the sum of the $n$ data points and the concentration parameter $\alpha$ of a Dirichlet distribution.  This sampling efficiency calculation directly simplifies to the known formula of the beta-binomial case from \cite{raiffa1961}(in our notation): $\eta = n / (\alpha_0 + \alpha_1)$ where $\pi \sim Beta(\alpha_0,\alpha_1)$.


Again, we make the notation more concrete, by revisiting our forecasting product demand example from the end of \S\ref{sec:dataLoss}.  Recall, we have a product that will sell between 0 and 5 units ($M=5$) and loss is scaled by $k=5$.  The decision maker is contemplating the value of $n=3$ observations.  Introducing a zero-inflated prior $\pi \sim Dirichlet(\tfrac{10}{6},\tfrac{1}{6},\tfrac{1}{6},\tfrac{1}{6},\tfrac{1}{6},\tfrac{1}{6})$ means $\alpha=\tfrac{15}{6}$, $c_1 = \tfrac{6}{15} \cdot (0 \cdot \tfrac{10}{6} + 1 \cdot \tfrac{1}{6} + 2 \cdot \tfrac{1}{6} + 3 \cdot \tfrac{1}{6} + 4 \cdot \tfrac{1}{6} + 5 \cdot \tfrac{1}{6}) = 1$, $c_2 = \tfrac{6}{15} \cdot (0 \cdot \tfrac{10}{6} + 1 \cdot \tfrac{1}{6} + 4 \cdot \tfrac{1}{6} + 9 \cdot \tfrac{1}{6} + 16 \cdot \tfrac{1}{6} + 25 \cdot \tfrac{1}{6}) = \tfrac{11}{3}$.  Plugging into eq. (\ref{eqn:EVSI}) yields EVSI $V_3(\pi)=\tfrac{160}{77} \approx 2.08$. and EVDI $V_{\infty}(\pi)=\tfrac{80}{21} \approx 3.81$.  From eq. (\ref{eqn:eta}) we get $\eta = \tfrac{6}{11} \approx 54.5\%$, so the learning from $n=3$ samples is expected to provide more than half of the maximum possible reduction in loss. Following from eqs. (\ref{eqn:beg}) - (\ref{eqn:end}), $a^*(\pi)=1$ and prior expected loss $\bar{R}(\pi,a^*(\pi)) = 5 \cdot (-1^2 \cdot \tfrac{10}{15} + 0^2 \cdot \tfrac{1}{15} + 1^2 \cdot \tfrac{1}{15} + 2^2 \cdot \tfrac{1}{15} + 3^2 \cdot \tfrac{1}{15} + 4^2 \cdot \tfrac{1}{15}) = \tfrac{40}{3} \approx 13.33.$  And thus, we can also get the prior expectation of posterior loss $\mathbb{E}_X\left[\bar{R}(\pi_X,a^*(\pi_X))\right] = \bar{R}(\pi,a^*(\pi)) - V_3(\pi) = \tfrac{40}{3} - \tfrac{160}{77} \approx 11.26.$

\section{Notes on Richness and Interpretability of Modelings Assumptions}
In the previous section, we showed one of the three EVSI desiderata, tractability, can be achieved for a multinomial data-generating process with Dirichlet prior.  The multinomial distribution is flexible enough to model any discrete (finite) data distribution.  Its prior, the Dirichlet distribution, is also flexible in its ability to model a wide range of distributions over a simplex.  Yet, some sacrifice of richness in modeling prior beliefs is made in the name of tractability.  Most notably, a more rich/flexible alternative prior over a simplex is the logistic-normal distribution \cite[see discussion in]{Aitchison1980}.  The  most glaring weakness of the Dirichlet distribution is in modeling prior beliefs where there is some type of correlation structure between data observations.  For example, observing a high data value, say 100,  would make one think values of 101 and 99 are also more likely to occur than data values further away.  However, the Dirichlet distribution, as a prior distribution to multinomial data, is unable to capture this structure.   Notably, the distribution-free underpinnings of the Kaplan-Meier estimator also ignore this potential correlation among data observations, yet shows favorable results in a similar repeated newsvendor setting \cite{huh2011} .

\begin{figure}
\centering
\includegraphics[width=0.95\textwidth]{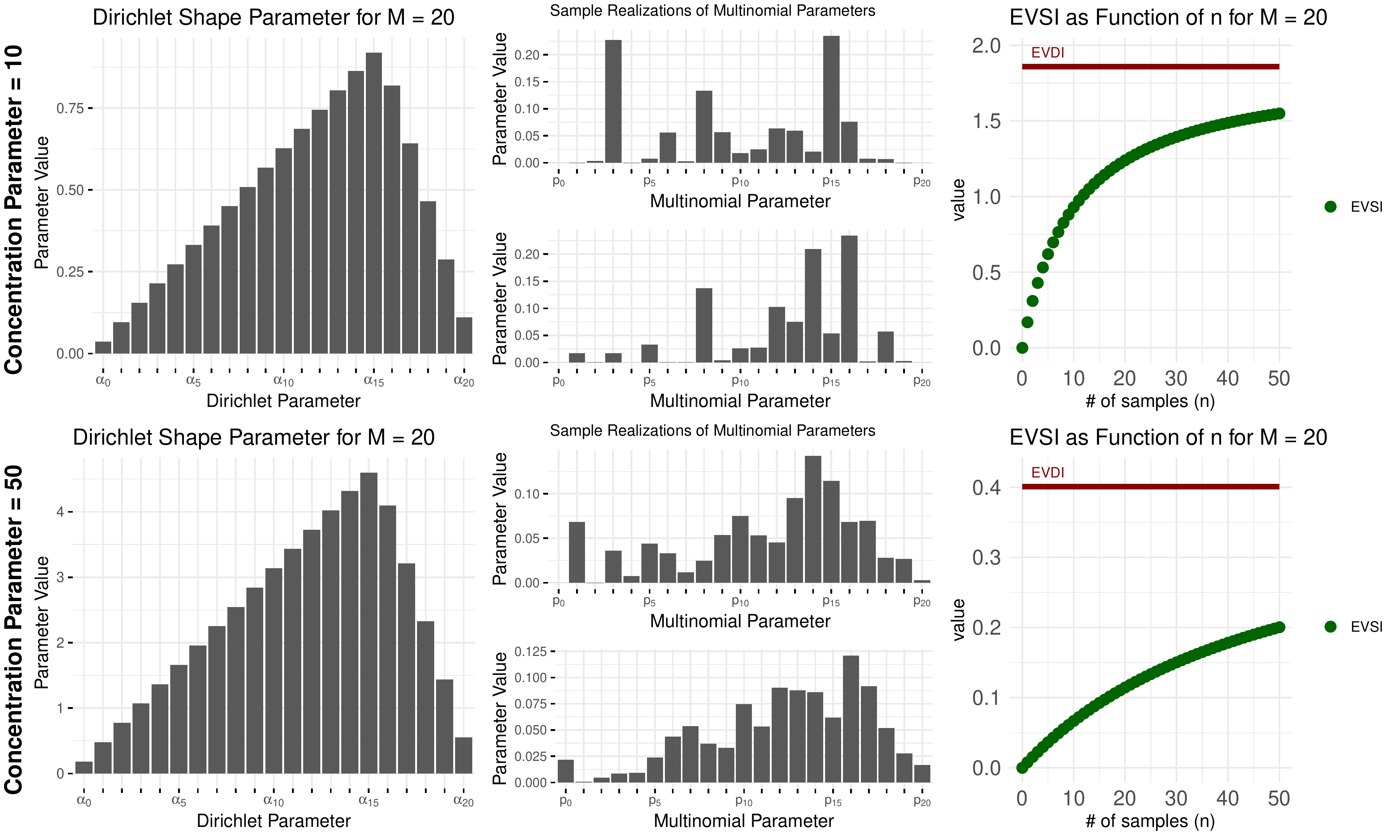}
\caption{Graphical depiction of the Dirichlet prior parameters, potential realizations for that prior (i.e. the multinomial parameters), and the EVSI/EVDI calculations as a function of $n$ samples for the given prior.  Top row for concentration parameter $\alpha = 10$ and bottom row for concentration parameter $\alpha = 50$}
\label{fig:evsiSkew}
\end{figure}

The richness of the Dirichlet prior is best seen through the lens of its intuitive reparameterization \cite{huang2005maximum}.  Let the concentration parameter $\alpha = \sum_{i=0}^M \alpha_i$ and let the vector $\mathbf{m} = \left(\frac{\alpha_0}{\alpha},  \frac{\alpha_1}{\alpha}, \ldots, \frac{\alpha_M}{\alpha} \right)$ represent the mean where the expected mean of the data observations is given as $c_{1} = \frac{1}{\alpha}\sum_{i=0}^M i \alpha_i = \sum_{i=0}^M i m_i$.  When $\alpha$ is small, say $\alpha \leq M$, the prior distribution over the simplex can differ greatly from $\mathbf{m}$ and reflect a decision maker's uncertainty around their expectation.  As $\alpha$ is made larger, the prior distribution will concentrate probability density near $\mathbf{m}$ and reflect greater confidence.  We present a graphical overview of this in Figure~\ref{fig:evsiSkew} for two different concentration parameters.  As seen, when $\alpha$ is smaller (top row of Figure~\ref{fig:evsiSkew}) the realized multinomial parameters (middle-top plot) can be further away from the mean $\mathbf{m}$ (which is proportional to the parameters in the top-left plot). As $\alpha$ increases (bottom-row) the prior distribution becomes much more informative and multinomial parameters will most likely mirror the prior Dirichlet parameters.

In terms of interpretability, Theorem \ref{th:evsi} formalizes our intuition about what drives the value of data.  Specifically, data is valuable when 1) the sample contains a lot of data (high $n$), 2) the expected variance of the data distribution is large (high $c_2 - c_1^2$), and 3) there is a lot of uncertainty regarding the true data distribution ($\alpha$ is small).  Additionally, the calculation for EVDI (eq. \ref{eqn:EVDI}) gives an interpretable upper bound on the value of data where high variance pushes to make samples more valuable and a high concentration parameter makes samples less valuable.  Lastly, the equation for efficiency (\ref{eqn:eta}) adds further insight by stating how quickly the upper bound on the value of data is approached; basically, the smaller the Dirichlet concentration parameter, the more quickly EVDI is approached with each subsequent data point.

\section{Illustrative Examples}
In this section, we demonstrate how the tractable formulation for EVSI, equation (\ref{eqn:EVSI}), can serve as a building block inside of other research initiatives.  The first example explores sample size optimization and the second example shows how a tractable EVSI calculation can lead to a tractable decision policy in a two-stage production planning problem.  In the third/last example, the EVSI formula provides a foundation from which to benchmark heuristic updating procedures that seek to estimate an underlying unknown data distribution.

\subsection{The Choice of Sample Size}
We now explore a decision maker's objective to choose the number of sample points to collect in such a way as to minimize his expected loss when assuming expected sampling cost, $C_s(n)$, is a linear function of the number of sampled points $n$: 
\begin{equation}
C_s(n) = K + s n \label{eqn:SampSize}
\end{equation}
where $s$ is the cost of one sample and $K$ represents the fixed costs of sampling.

The loss function to be minimized, $\ell_s(n)$,  combines equations (\ref{eqn:EVSI}) and (\ref{eqn:SampSize}):
\begin{equation}
\ell_s(n) = -\frac{k n (c_2-c_1^2) } {(n + \alpha)(1+\alpha)} + K + s n
\end{equation}

And assuming for practical purposes that $n$ can be treated continuously, we get the optimal sample size:
\begin{equation}
n^* = \sqrt{\frac{\alpha}{(1 + \alpha)} \frac{k}{s}(c_2 - c_1^2)} - \alpha \label{eqn:optSampSize}
\end{equation}
for cases where $n^*$ is positively valued and the fixed costs of sampling $K$ can be recovered, i.e. $V_n(\pi) > C_s(n^*$).  In all other cases, $n^* = 0$.  Equation (\ref{eqn:optSampSize}) has a nice economic interpretation where the three terms represent the strength of the prior, the ratio between the scaling of the quadratic loss costs and the unit sampling costs, and the predicted variance of the data distribution.

\subsection{Two-Stage Production Planning}
The example shown here is a simple two-stage production planning problem (see, e.g., \cite{fisher2001}) where the decision maker seeks to optimally schedule the $2^{nd}$ production run. 

Assume $J$ periods make up a selling season.  Each period, $j \in J$ faces independent and identical categorical demand with Dirichlet prior and quadratic loss (i.e. a repeated newsvendor setting with quadratic loss) with identical shipments scheduled for each period.  A decision maker can choose either 1) to schedule the delivery quantity for each period in the entire selling season or, 2) at cost $K$ can specify a period $j^*$ after which the scheduled delivery quantity can be changed.  Assuming this change date will be contractually set in advance of the selling season, find $j^*$ to minimize expected net costs over the entire season $J$.

The net cost function for this problem is:
\begin{equation}
  C(j)= \begin{cases}
         0, & \text{if $j=0$},\\
         K - \left(J-j\right)\dfrac{kj(c_2 - c_1^2)}{(j + \alpha)(1+\alpha)}, & \text{if $j\in (0,J]$}
        \end{cases}
\end{equation}

When $j\in (0,J]$, the net cost function $C(\cdot)$ is strictly convex and has a unique global minimum value. The optimal period $j^*$ is
\begin{equation*}
  j^* = \underset{j \in {\{0,1,\ldots, J\}}}{\arg \min} ~C(j)
\end{equation*}


When $\min C(j) = 0$ for $ 0<j\leq J$, we choose $j^* = 0$.

For the case when $\min C(j)<0$, we have
\begin{equation*}
    j^* = \sqrt{\alpha(J+\alpha)}-\alpha
\end{equation*}

Considering that $j^*$ must be a non-negative integer, summarizing different cases we have the optimal $j^*$ as
\begin{equation}
   j^* = \begin{cases}
         0, & \text{if $\underset{j \in [0, J]} {\min} C(j)=0$},\\
         \underset{j \in {\{\lfloor{j_0}\rfloor, \lceil{j_0}\rceil\}}}{\arg \min} ~C(j), & \text{if $\underset{j \in [0, J]} {\min} C(j)<0$}.
         \end{cases}
\end{equation}
where $j_0 = \sqrt{\alpha(J+\alpha)}-\alpha$.

\subsection{Benchmarking Data-Driven Algorithms}
An active area of research is to propose algorithms for decisions in repeated settings where minimal assumptions about the underlying data distribution are known.  These approaches include Sample Average Approximation(SAA) \cite{Retsef2007,levi2015data}, concave adaptive value estimation (CAVE) \cite{Godfrey2001}, and Second Order Belief Maximum Entropy (SOBME) \cite{saghafian2016newsvendor}.  When benchmarking these algorithms, it is customary to pick a handful of ``true'' distributions where the algorithm competes against a known optimal solution.

With the introduction of a closed-form EVSI calculation in the context of a Dirichlet prior, a more robust benchmarking scenario can be achieved.  Instead of picking a ``true'' data distribution, we pick a ``true prior'' from the Dirichlet family with support matching the problem of interest.  This prior can be used to then simulate ``true'' data distributions (as many as we want) by which we can estimate the reduction in squared loss as a function of $n$, the number of data samples.  Given this setup, a comparison of a proposed algorithm can be made against a known optimal \textit{updating procedure}.  After all, it is the updating procedure that we are seeking to validate, and the optimal updating procedure to benchmark new algorithms against is, therefore, the Bayesian one detailed in the proof of Theorem \ref{th:evsi} (see appendix).

\begin{figure}
\centering
\includegraphics[width=0.6\textwidth]{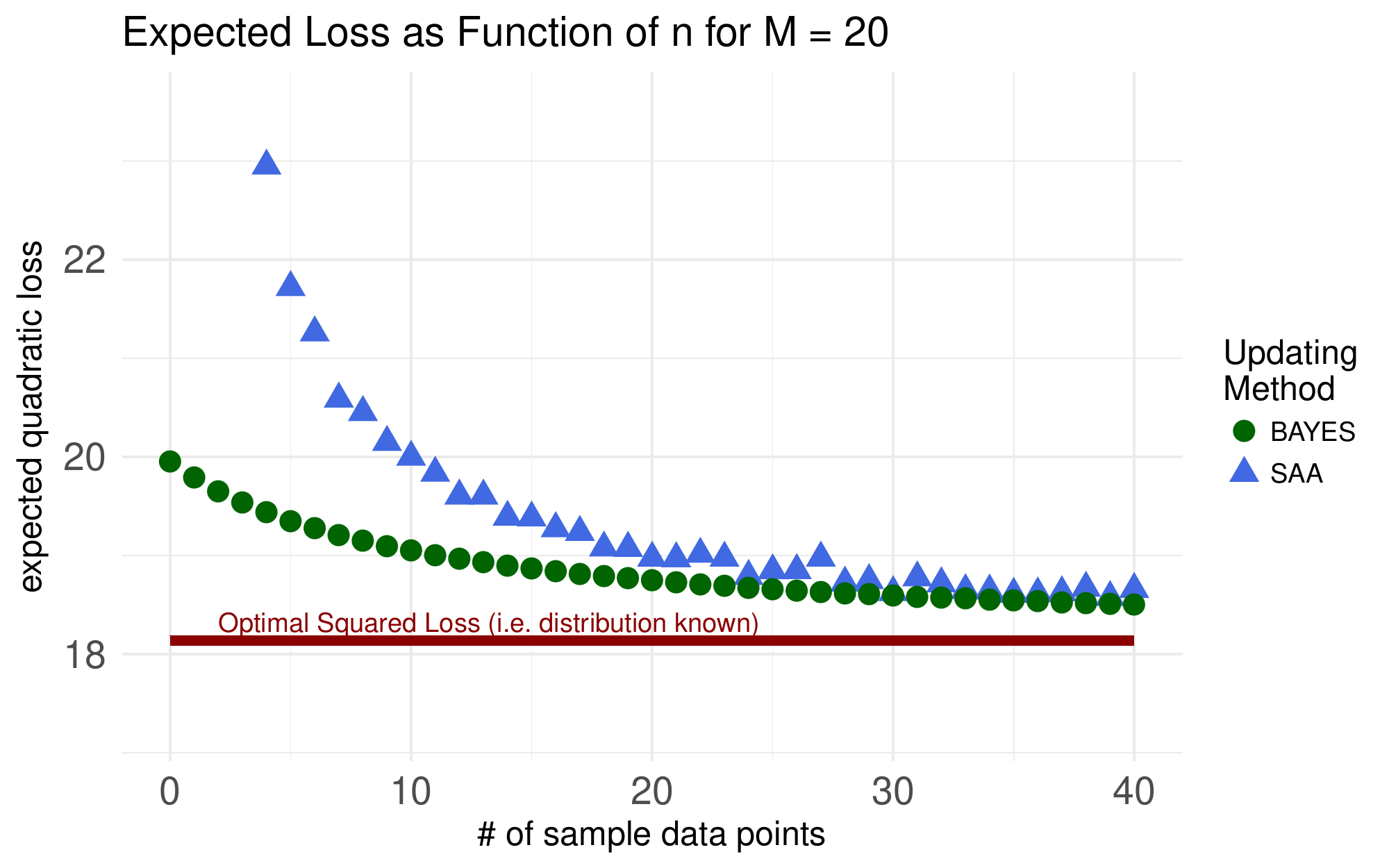}
\caption{Comparing the sample average approximation(SAA) updating procedure to the known Bayesian (BAYES) optimal updating procedure.}
\label{fig:benchmark}
\end{figure}

As a proof of concept, Figure~\ref{fig:benchmark} is an example benchmarking the well-known sample average approximation (SAA) (see \cite{Retsef2007}) against the known optimal Bayesian updating procedure (BAYES) using a $\textrm{Dirichlet}(\alpha_0, \alpha_1, \ldots, \alpha_M)$ prior with $M=20$, $\alpha = 10$, and $\mathbf{m} \propto  \{1,2,3,4,5,6,7,8,9,10,11,12,13,14,13,11,9,7,5,3,1\}$ (chosen to be slightly skewed). In this scenario, we see the value of prior information in small data settings as BAYES outperforms SAA.  It also shows how as the amount of data increases, the non-parametric SAA algorithm's performance improves and closely mimics that of the optimal Bayesian updating procedure.

\section{Conclusion}
The use of preposterior analysis in this paper provides a formal method for valuing data prior to its collection and as such, should serve as a building block in many systems and models going forward.  By expanding the support of the underlying data-generating process from $[M] = \{0,1\}$ to $[M] = \{0,1,\ldots,M\}$, the beta-binomial EVSI calculations are successfully generalized to a Dirichlet-multinomial setting.  Using this new EVSI computation, three illustrative examples valuing data prior to its collection are shown, there are potentially many other contexts where this tractable formulation might also prove useful. Researchers in two particular areas, medical decision making and active (machine) learning are known to be interested in EVSI types of calculations (see, e.g., \cite{ades2004expected,haertel2008return,jackson2019value,muesing2021fully}). And we look forward to hearing of other useful deployments for this method of valuing data prior to its collection.

%
\appendix

\section{Proof of Proposition \ref{th:nonnegativity} and Theorem \ref{th:evsi}}
\label{appendix}

\subsection{Proof of Proposition \ref{th:nonnegativity}}

The expected value of sample information is 
\begin{equation}
\label{proof: Vbar}
  \begin{split}
  V_n\left(\pi\right)
  &= \mathbb{E}_{T} \left[R(T,a^*(\pi))\right] - \mathbb{E}_T \mathbb{E}_{X|T} \left[R(T,a^*(\pi_X)) \right].
  \end{split}
\end{equation}

For the first term in eq. (\ref{proof: Vbar}), we have
  \begin{equation}
  \label{eqn:avgBayesRiskPi}
    \begin{split}
      \mathbb{E}_{T} \left[R(T,a^*(\pi))\right]
      &= k \mathbb{E}_T\left[ \mathbb{E}_{D|T} \left[ \left(D-a^*\left( \pi \right) \right)^2 \right] \right] \\
      &= k \mathbb{E}_T\left[ \mathbb{E}_{D|T} \left[ \left(D-\mathbb{E}\left[D\right] \right)^2 \right] \right]\\
      &= k\mathbb{E}_D\left[ \left( D - \mathbb{E}\left[D\right] \right)^2\right]\\
      &= k\textrm{Var}\left[D\right].
    \end{split}
  \end{equation}
The second line is due to the optimal action under squared loss being the mean (see eq. (\ref{eqn:expD})).  The third line of equation (\ref{eqn:avgBayesRiskPi}) follows from the \emph{law of total expectation}. Thus, the optimal Bayes risk without sample information under quadratic loss (\ref{eq:quadLoss}) is the marginal variance of $D$ scaled by a factor $k$.

Similarly, for the second term in eq. (\ref{proof: Vbar}) we find
\begin{equation}
\label{eqn:avgBayesRiskPi_x}
  \begin{split}
     \mathbb{E}_T \left\{ \mathbb{E}_{X|T} \left[ R\left( T, a^*\left(\pi_X\right)\right) \right]\right\}
     &= k\mathbb{E}_T \left\{ \mathbb{E}_{X|T} \left[ \mathbb{E}_{D|T} \left[ \left(D-a^*\left( \pi_X \right) \right)^2 \right]\right]\right\} \\
     &= k\mathbb{E}_T \left\{ \mathbb{E}_{X|T} \left[ \mathbb{E}_{D|T} \left[ \left(D - \mathbb{E}_{D|X}\left[ D \right] \right)^2 \right] \right]\right\} \\
     &= k \mathbb{E}_X \left\{ \mathbb{E}_{D|X} \left[ \left( D - \mathbb{E}_{D|X} \left[ D \right] \right)^2 \right] \right\}\\
     &= k \mathbb{E}_X \left[ \textrm{Var}_{D|X}\left[D\right] \right].
  \end{split}
\end{equation}
The optimal Bayes risk under quadratic loss (\ref{eq:quadLoss}) if a sample of size $n$ is to be collected is the expected variance of the predictive posterior distribution of $D$ scaled by a factor $k$.

Combining (\ref{proof: Vbar}),(\ref{eqn:avgBayesRiskPi}), and (\ref{eqn:avgBayesRiskPi_x}), we complete the proof:
\begin{equation}
\label{proof: Vbar2}
  \begin{split}
  V_n\left(\pi\right)
  &= \mathbb{E}_{T} \left[R(T,a^*(\pi))\right] - \mathbb{E}_T \left\{ \mathbb{E}_{X|T} \left[R\left( T, a^*\left(\pi_X\right)\right) \right]\right\} \\
  &= k \textrm{Var}\left[D\right] - k\mathbb{E}_X \left[ {\textrm{Var}_{D|X} \left[D \right]} \right] \\
  &= k \left\{ {\textrm{Var}\left[D\right]} - {\mathbb{E}_X \left[ {\textrm{Var}_{D|X} \left[D \right]} \right]} \right\} \\
  &= k \textrm{Var}_X\left[ {\mathbb{E}_{D|X}\left[ D \right]} \right] \\
  &\geq 0.
  \end{split}
\end{equation}
The last equal sign in equation (\ref{proof: Vbar2}) follows from the \emph{law of total variance}.
Since $k>0$ and $\textrm{Var}_X\left[ {\mathbb{E}_{D|X}\left[ D \right]} \right] \geq 0$ for any $X$, we have $V_n\left(\pi\right)\geq 0$ for any sample size $n$.
\hfill$\Box$

\subsection{Proof of Theorem \ref{th:evsi}}

Consider the prior distribution for the data-generating process
\begin{equation*}
\pi = \textrm{Dirichlet}(\alpha_0,\alpha_1,\ldots,\alpha_M).
\end{equation*}
Suppose our information consists of $n$ samples of the data distribution.
Let $n_j$, $j \in [M]$ be the frequency of the data
being $j$ so that $n_j$ are integers such that $\sum_{j=0}^M n_j = n$.
Then, because the multinomial and Dirichlet distributions are conjugate,
\begin{eqnarray*}
\pi_X &=& \textrm{Dirichlet}(\alpha_0+n_0,\alpha_1+n_1,\ldots,\alpha_M+n_M).
\end{eqnarray*}
Because $\pi$ and $\pi_X$ both belong to the same class of distributions, we derive closed-form
valuations for the information $X$. The corresponding marginal likelihoods for $\pi$ and $\pi_X$ are
\begin{eqnarray*}
q_{\pi}(d) &=& \frac{\alpha_d}{\alpha}, \\
q_{\pi_X}(d) &=& \frac{\alpha_d+n_d}{\alpha +n},
\end{eqnarray*}
where $\alpha = \sum_{i=0}^M \alpha_i$.
If the information happens to occur in such a way that $n_j \propto \alpha_j$ for each $j$, then
the updated marginal likelihood is unchanged: $q_d(\pi) = q_d(\pi_X)$, $d \in [M]$.

For convenience, define the quantities
\begin{eqnarray*}
Z &=& \frac{1}{n} \sum_{d=0}^M d n_d,\\
c_{1} &=& \frac{1}{\alpha}\sum_{d=0}^M d \alpha_d, \\
c_{2} &=& \frac{1}{\alpha}\sum_{d=0}^M d^2 \alpha_d,
\end{eqnarray*}
where $Z$ represents the average frequency of the sample, $c_1$ the prior expectation for a sample value, and $c_2$ the prior second moment for the sample value. 

Given the loss function in (\ref{eq:quadLoss}), the Bayes risk and action without sample information can be explicitly calculated
\begin{eqnarray}
\bar{R}(\pi,a) &=& \mathbb{E}_{T\sim\pi}[R(T,a)], \label{eqn:beg} \\
&=& \mathbb{E}_{T\sim \pi}\left[ \sum_{d=0}^M p_T(d) \ell(d,a) \right], \\
&=& \sum_{d=0}^M \ell(d,a) \mathbb{E}_{T\sim \pi}[p_T(d)], \\
&=& \sum_{d=0}^M  \ell(d,a) q_{\pi}(d), \label{eqn:qd}
\end{eqnarray}
where $\{q_{\pi}(0),q_{\pi}(1),\ldots,q_{\pi}(M)\}$ is the \emph{marginal likelihood}.
The Bayes action minimizes eq. (\ref{eqn:qd}):
\begin{eqnarray}
\nonumber \pt{\bar{R}(\pi,a)}{a} &=& -2k\sum_{d=0}^M (d-a) q_{\pi}(d) = -2k \left(\sum_{d=0}^M d q_{\pi}(d) - a \sum_{d=0}^M q_{\pi}(d)\right) = 0,\\
\nonumber \Rightarrow a^*(\pi) &=& \sum_{d=0}^M d q_{\pi}(d),\\
&=& \mathbb{E}_{q_{\pi}}[D], \label{eqn:expD}\\
&=& c_1.  \label{eqn:end}
\end{eqnarray}
the mean data outcome under the prior marginal likelihood. The corresponding Bayes Risk is
\begin{eqnarray*}
\bar{R}(\pi,a^*(\pi)) &=& k\sum_{d=0}^M(d-a^*(\pi))^2 q_{\pi}(d),\\
&=& k\textrm{Var}_{q_{\pi}}[D], \\
&=& k(c_2 - c_1^2).
\end{eqnarray*}
Similarly, with sample information we have
\begin{eqnarray}
\nonumber \pt{\bar{R}(\pi_X,a)}{a} &=& -2k\sum_{d=0}^M (d-a) q_{\pi_X}(d),\\
\nonumber &=& -2 k\left(\sum_{d=0}^M d q_{\pi_X}(d) - a \sum_{d=0}^M q_{\pi_X}(d)\right) = 0,\\
\nonumber \Rightarrow a^*(\pi_X) &=& \sum_{d=0}^M d q_{\pi_X}(d),\\
\nonumber &=& \mathbb{E}_{q_{\pi_X}}[D], \\
&=& \frac{\alpha c_1 + nZ}{\alpha+n}, \label{eqn:astar}
\end{eqnarray}
which is the mean data outcome under the posterior marginal likelihood. Now expressing EVSI as
\begin{eqnarray}
V_n(\pi) &=& \bar{R}(\pi,a^*(\pi)) - \mathbb{E}_{T} \mathbb{E}_{X|T} R(T,a^*(\pi_X)), \label{eqn:VV}
\end{eqnarray}
note the inner expectation is taken over the data frequency, which follows
a multinomial distribution: $(n_0,\ldots,n_M) \sim \textrm{Multinomial}(p_t(0),\ldots,p_t(M))$),
and the outer expectation is taken over all possible distributions $p_{t^*} \sim \textrm{Dir}(\alpha_0,\ldots,\alpha_M)$.

The first term in (\ref{eqn:VV}) has already been evaluated as $k(c_2 - c_1^2)$.
We now calculate the second term.
\begin{eqnarray}
\nonumber R(t,a^*(\pi_X)) &=& k\sum_{d=0}^M p_t(d) (d-a^*(\pi_X))^2 \\
\nonumber &=& k\sum_{d=0}^M p_t(d) \left[d-\frac{\alpha c_1 + nZ}{\alpha+n}\right]^2 \\
\nonumber \Rightarrow \mathbb{E}_{X|T=t}\left[ R(t,a^*(\pi_X)) \right]
&=& k\sum_{d=0}^M p_t(d) \left[ d^2 - \left( \frac{2nd}{\alpha+n} - \frac{2n\alpha c_1}{(\alpha+n)^2}\right) \right.
\mathbb{E}_X[Z] \notag \\
&& \left. - \frac{2d\alpha c_1}{\alpha+n} + \frac{\alpha^2 c_1^2}{(\alpha+n)^2}
+ \frac{n^2}{(\alpha+n)^2}\mathbb{E}_X[Z^2]\right]. \label{eqn:1a}
\end{eqnarray}
Since $Z(n_0,\ldots,n_M) = \frac{1}{n}\sum_{d=0}^M d n_d$,
\begin{eqnarray*}
\mathbb{E}_{X|T=t}[Z] &=& \sum_{d=0}^M d p_t(d), \\ \label{EZ2}
\mathbb{E}_{X|T=t}[Z^2] &=& \textrm{Var}_{X|T=t}[Z] + \left(\mathbb{E}_{X|T=t}[Z]\right)^2 \\
              &=& \frac{1}{n}\sum_{d=0}^M d^2 p_t(d)  + \frac{(n-1)}{n}\left( \sum_{d=0}^M d p_t(d) \right)^2,
\end{eqnarray*}
where the last line follows from the fact
\begin{align}
	\textrm{Var}_{X|T=t}[Z] 
	&= \textrm{Var}_{X|T=t}\left[\frac{1}{n}\sum_{d=0}^M d n_d\right] \nonumber \\
	&= \frac{1}{n^2}\textrm{Var}_{X|T=t}\left[\sum_{d=0}^M d n_d\right] \nonumber \\
	&= \frac{1}{n^2}\left\{\sum_{d=0}^M d^2 \textrm{Var}_{X|T=t}\left[n_d\right] + 
	    2 \sum_{0=i<j}^M ij \textrm{Cov}_{X|T=t}\left(n_i, n_j\right) \right\} \nonumber \\
	&= \frac{1}{n^2}\left\{\sum_{d=0}^M d^2 n p_t(d)\left(1 - p_t(d)\right) -
	    2\sum_{0=i<j}^M ij np_t(i)p_t(j) \right\} \nonumber \\
	&= \frac{1}{n}\sum_{d=0}^M d^2p_t(d) - 
	    \frac{1}{n} \left\{ \sum_{d=0}^M d^2 p_t^2(d) + 2 \sum_{0=i<j}^M ijp_t(i)p_t(j)\right\} \nonumber \\
	&= \frac{1}{n}\sum_{d=0}^M d^2p_t(d) - \frac{1}{n}\left(\sum_{d=0}^M d p_t(d) \right)^2.
\end{align}
Eq. (\ref{eqn:1a}) becomes
\begin{eqnarray*}
\mathbb{E}_{X|T=t}[R(t,a^*(\pi_X))] &=&
k \left\{ \left(1 + \frac{n}{(\alpha+n)^2}\right) \sum_{d=0}^M d^2 p_t(d) +
\left( \frac{2\alpha n c_1}{(\alpha+n)^2} - \frac{2\alpha c_1}{\alpha+n} \right) \sum_{d=0}^M d p_t(d) \right.\\
&& \left.+ \left[ \frac{n(n-1)}{(\alpha+n)^2} - \frac{2n}{\alpha+n} \right]\left[ \sum_{d=0}^M d p_t(d) \right]^2 + \frac{\alpha^2 c_1^2}{(\alpha+n)^2}\right\}.
\end{eqnarray*}
The final step is to take the expectation over all possible beliefs $p_t \sim
\textrm{Dirichlet}(\alpha_0,\ldots,\alpha_M)$. Using the fact that 
\begin{eqnarray*}
\mathbb{E}_{T\sim \pi}[p_T(i)] &=& \frac{\alpha_i}{\alpha}, \\
\mathbb{E}_{T\sim \pi}[p_T(i)^2] &=& \textrm{Var}[p_T(i)] + \mathbb{E}_T[p_T(i)]^2 \\
&=& \frac{\alpha_i(\alpha-\alpha_i)}{\alpha^2(\alpha+1)} + \frac{\alpha_i^2}{\alpha^2} \\
&=& \frac{\alpha_i(\alpha_i+1)}{\alpha(\alpha+1)}, \\
\mathbb{E}_{T\sim \pi}[p_T(i) p_T(j)]  &=& \textrm{Cov}[p_T(i),p_T(j)] + \mathbb{E}_T[p_T(i)] \mathbb{E}_T[p_T(j)], \qquad i \neq j, \\
&=& -\frac{\alpha_i \alpha_j}{\alpha^2(\alpha+1)} + \frac{\alpha_i \alpha_j}{\alpha^2} \\
&=& \frac{\alpha_i \alpha_j}{\alpha(\alpha+1)},
\end{eqnarray*}
and
\begin{align}
	\mathbb{E}_{T\sim \pi}\left[\left( \sum_{d=0}^M d p_T(d) \right)^2\right]
	&= \mathbb{E}_{T\sim \pi}\left[\sum_{d=0}^M d^2 p_T^2(d) + 2\sum_{0=i<j}^M i j p_T(i) p_T(j)\right] \nonumber \\
	&= \sum_{d=0}^M d^2 \mathbb{E}_{T\sim \pi}\left[ p_T^2(d) \right] + 2\sum_{0=i<j}^M i j \mathbb{E}_{T\sim \pi}\left[ p_T(i) p_T(j) \right] \nonumber \\
	&= \sum_{d=0}^M d^2 \frac{\alpha_d \left(\alpha_d + 1\right)}{\alpha \left(\alpha + 1 \right)} +
	    2\sum_{0=i<j}^M i j \frac{\alpha_i \alpha_j}{\alpha \left(\alpha + 1\right)} \nonumber \\
	&= \frac{1}{\alpha + 1}\sum_{d=0}^M \frac{d^2 \alpha_d}{\alpha } + 
	    \frac{1}{\alpha \left(\alpha + 1\right)} \left(\sum_{d=0}^M d^2 \alpha_d^2 + 2\sum_{0=i<j}^M i j  \alpha_i \alpha_j  \right) \nonumber \\
	&= \frac{c_2}{\alpha + 1} + \frac{1}{\alpha \left(\alpha + 1\right)} \left(\sum_{d = 0}^Md \alpha_d \right)^2 \nonumber \\
	&= \frac{c_2}{\alpha + 1} + \frac{\alpha c_1^2}{\alpha + 1},
\end{align}
we obtain
\begin{eqnarray*}
\mathbb{E}_T \mathbb{E}_{X|T}[R(T,a^*(\pi_X))]
&=& k \left\{ \left(1 + \frac{n}{(\alpha+n)^2}\right) \sum_{d=0}^M  \frac{d^2 \alpha_d}{\alpha} +
\left( \frac{2\alpha n c_1}{(\alpha+n)^2} - \frac{2\alpha c_1}{\alpha+n} \right) \sum_{d=0}^M  \frac{d\alpha_d}{\alpha} \right. \notag \\
&& \left.+ \left[ \frac{n(n-1)}{(\alpha+n)^2} - \frac{2n}{\alpha+n} \right]\left[
\frac{c_2}{\alpha+1} + \frac{\alpha c_1^2}{\alpha+1}
\right] + \frac{\alpha^2 c_1^2}{(\alpha+n)^2} \right\}\\
&=& k(c_2 - c_1^2) \frac{\alpha(1+\alpha+n)}{(1+\alpha)(n+\alpha)}.
\end{eqnarray*}
The value of $n$ samples from the data distribution is therefore
\begin{eqnarray}
V_n(\pi) &=& k(c_2 - c_1^2) - k(c_2 - c_1^2) \frac{\alpha(1+\alpha+n)}{(1+\alpha)(n+\alpha)} \nonumber \\
&=& \frac{k n (c_2-c_1^2) }{(n + \alpha)(1+\alpha)}.
\label{eqn:2}
\end{eqnarray}
\hfill$\Box$

%



\end{document}